\documentclass[12pt, a4paper]{article}

\usepackage{a4wide}
\usepackage[T2A]{fontenc}
\usepackage[utf8]{inputenc}
\usepackage{amsfonts}
\usepackage{amssymb}
\usepackage{amsthm}
\usepackage{amsmath}
\usepackage{mathtools}
\usepackage{comment}
\usepackage{color}

\usepackage[hidelinks]{hyperref}

\usepackage[pdftex]{graphicx}

\DeclarePairedDelimiter{\prt}{\lfloor}{\rfloor}

\begin{document}

\def\newstr{\par\noindent}
\def\ms{\medskip\newstr}

\renewcommand{\proofname}{Proof}
\renewcommand{\labelenumi}{$\bullet$}
\newcommand{\pluseq}{\mathrel{+}=}

\newcommand{\ifdraft}[1]{#1}
\definecolor{aocolour}{rgb}{0.7,0.8,1}
\definecolor{vmcolour}{rgb}{1,0.8,0.7}
\newcommand{\ao}[1]{\ifdraft{\noindent\colorbox{aocolour}{A.O.: #1}}}
\newcommand{\vm}[1]{\ifdraft{\noindent\colorbox{vmcolour}{V.M.: #1}}}

\newcommand{\Z}{\mathbb{Z}}
\newcommand{\N}{\mathbb{N}}
\newcommand{\R}{\mathbb{R}}
\newcommand{\Q}{\mathbb{Q}}
\newcommand{\K}{\mathbb{K}}
\newcommand{\Cm}{\mathbb{C}}
\newcommand{\Pm}{\mathbb{P}}
\newcommand{\Zero}{\mathbb{O}}
\newcommand{\F}{\mathbb{F}_2}
\newcommand{\ilim}{\int\limits}
\newcommand{\impl}{\Rightarrow}
\newcommand{\set}[2]{\{ \, #1 \mid #2 \, \}}

\newcommand{\A}{\mathcal{A}}
\newcommand{\B}{\mathcal{B}}
\renewcommand{\C}{\mathcal{C}}

\theoremstyle{plain}
\newtheorem{theorem}{Theorem}
\newtheorem{lemma}{Lemma}
\newtheorem{st}{Statement}
\newtheorem{prop}{Properties}

\theoremstyle{definition}
\newtheorem{definition}{Definition}
\newtheorem{example}{Example}
\newtheorem{cor}{Corollary}
\newtheorem{problem}{Problem}

\newtheorem{oldtheorem}{Theorem}
\renewcommand{\theoldtheorem}{\Alph{oldtheorem}}

\theoremstyle{remark}
\newtheorem{remark}{Remark}
\newtheorem{note}{Note}

\author{Vladislav Makarov\thanks{Saint-Petersburg State University. Supported by Russian Science Foundation, project 18-11-00100. }}
\title{Counting ternary square-free words quickly}

\maketitle

\begin{abstract}
An efficient, when compared to exhaustive enumeration,  algorithm for computing the number of square-free words of length $n$ over the alphabet $\{a, b, c\}$ is presented.
\end{abstract}
 
\section{Introduction}\label{intro}

A \emph{square} is a string of form $ww$ for some non-empty string $w$. 
A string is \emph{square-free}, if it has no square substrings.
Over the binary alphabet, there is only finite number of square-free strings,
but the number of square-free strings of length $n$ over the ternary alphabet grows with $n$ slowly,
but exponentially~\cite{expo}.

Since the dawn of combinatorics of words, there has been a lot of research on bounding the number
of ternary square free words of length $n$ from below and from above (OEIS sequence A006156~\cite{oeis}). 
See, for example, a classic
review by Berstel~\cite{rev} and a new review by Shur~\cite{shur}, to see how much the state
of art has changed in-between.

Most of the research was focused on the estimating the numbers from above and from below,
culminating in Kolpakov's~\cite{lower} and Shur's~\cite{shur} methods of proving lower and upper
bounds on the growth rate of ternary square-free words (and other power-avoiding words
as well), that can be made as close as needed, given enough computational resources.

Computing their \emph{exact} number has attracted significantly less attention.
Back in 2001, Grimm~\cite{grimm} obtained the desired values up to $110$, but mostly
in order to prove a new upper bound on their growth rate.

We will go up to $n = 141$ on a completely ordinary laptop. In just a few hours. 
This paper gives a high-level account of underlying ideas, for
implementation details and possible optimisations refer 
to the repository with implementation~\cite{github}.

Most ideas used here are very classic, especially using the \emph{antidictionaries},
consisting of minimal squares and building an Aho-Corasick automaton for them. The key
contributions are Lemmas~\ref{square-overlap} and~\ref{key-lemma} and the way they are used in the final algorithm.

As usual, $O^*$ notation suppresses polynomial factors, so $O^*(1)$ is any at most polynomially
growing function, $O^*(2^n)$ is $O(2^n \cdot \mathrm{poly}(n))$, et cetera.

\section{A simple (but mostly useless) algorithm}

Unless the opposite is explicitly mentioned,
all strings in the following text are over the alphabet $\Sigma = \{a, b, c\}$.

It is not necessary to know how Aho-Corasick automaton works in order to understand
this paper.
The only important part is the following theorem, which is a simple
consequence of a more general result by Aho and Corasick:

\begin{oldtheorem}[Aho and Corasick, 1975~\cite{karas}] \label{crucian} For any finite subset $S$ of $\Sigma^*$, 
there exists a deterministic finite automaton $A$ with at most 
$1 + \sum\limits_{w \in S} |w|$ states, such that $L(A)$ is exactly the language of
all strings that contain at least one string from $S$ as a substring. Moreover, such
an automaton can be constructed in $O(\sum\limits_{w \in S} |w|)$ time.
\end{oldtheorem}

\begin{definition} A string is a \emph{minimal square} if it is a square, but does not contain 
any smaller squares as substrings. 
\end{definition}

\begin{definition} Denote the set of all square-free strings of length exactly $\ell$ by $L_\ell$.
\end{definition}

\begin{definition} Similarly, denote the set of all minimal squares with half-length at most $\ell$ by 
$M_\ell$.
\end{definition}

The main problem at hand is computing $|L_n|$. A string of length $n$ is square-free
if and only if it does not contain any \emph{minimal} squares of half-length at most $\prt{n/2}$. 
Indeed, if a string has a non-minimal square substring, it has a smaller square substring
by definition of non-minimality.

Let $A = (Q, q_0, \delta, F)$ be a DFA from Theorem~\ref{crucian} for the set $M_{\prt{n/2}}$.
Here, and in the rest of the text, $Q$ is the set of states, $q_0$ is the starting state, $\delta \colon
Q \times \Sigma \to Q$ is the transition function (with the usual extension to function $Q \times \Sigma^* \to Q$) and $F$ is the set of accepting states.

Let $f(\ell, q)$ be the number of strings $w \in \Sigma^\ell$, such that $\delta(q_0, w) = q$.
Then, we can compute all values of $f(\ell, q)$ row-by-row. Indeed, we know $f(0, \cdot)$:
$f(0, q_0) = 1$ and $f(0, q) = 0$ for $q \neq q_0$. Moreover, we can compute $f(\ell + 1, \cdot)$
through $f(\ell, \cdot)$: to compute $f(\ell + 1, q)$, sum up $f(\ell, p)$ over all predecessors of
$q$. In other words, $f(\ell + 1, q) = \sum_{p \in Q, d \in \Sigma, \delta(p, d) = q} f(\ell, p)$. 

In the end, $\sum\limits_{q \in F} f(n, q)$ is the number of strings of length $n$ that are \emph{not}
square free. So, to compute $|L_n|$, it is enough to:
\begin{enumerate}
\item[1)] Find the set $M_{\prt{n/2}}$ of all short minimal squares.
\item[2)] Build the automaton $A$ from Theorem~\ref{crucian} for $M_{\prt{n/2}}$.
\item[3)] Compute $f(n, q)$ for all $q \in Q$.
\end{enumerate}
How to do all these things?
\begin{enumerate}
\item[1)] Iterate over all square-free words of length at most $\prt{n/2}$ in $O^*(L_{\prt{n/2}})$ time and $O^*(1)$ memory, and, for each of them, check whether it is a minimal square when doubled. It is possible to
achieve a polynomial in $n$ speed-up by building the automaton from point 2 on the fly~\cite[Subsection 3.3]{shur}
or with other similar optimisations. However, the speed-up would still be only polynomial.
\item[2)] Just use Theorem~\ref{crucian}, the resulting automaton will have at most $2\prt{n/2} M_{\prt{n/2}} + 1 = O^* (M_{\prt{n/2}})$ states and can be constructed in $O^*(M_{\prt{n/2}})$
time.
\item[3)] Compute the values of $f(\ell, \cdot)$ row-by-row, using above formulas. It is enough
to keep only values of $f(\ell, \cdot)$ and $f(\ell + 1, \cdot)$ in the memory. Here, we need
$O^*(M_{\prt{n/2}})$ of both time and memory.
\end{enumerate}

In total, we need $O^*(L_{\prt{n/2}})$ time and $O^*(M_{\prt{n/2}})$ memory, as promised.

\section{An improved algorithm}\label{main_algo}

The main factor that limits the practical usefulness of the above approach 
 is the memory usage (compared to the high running time of a naive algorithm).

So, we want to reduce the memory consumption, possibly at the cost of making running
time slightly worse.
The key observation that makes this possible is a pretty interesting one 
and can be seen as an incomplete application of the inclusion-exclusion principle. 

Intuitively, the knowledge that some fixed substring of a string is a square, gives a lot
of constraints of type ``symbols on some positions are equal''. Hence, two long (whatever
that means, the exact definition of ``long'' will come later) square
strings place too many constraints on the string, meaning that there are a lot of 
symbols that are ``forced'' to be equal, hence we can find a smaller square substring.
 
\begin{example} Suppose that have a string $s$ of length $13$ and we know that its prefix of length
$8$ is a square, and so is its suffix of length $12$. Hence, we know that $s_i = s_{i + 4}$ for $0 \leqslant i < 4$ and, similarly, $s_i = s_{i + 6}$ for $1 \leqslant i < 7$. Then, $s$ looks like $1232123232123$, where equal digits correspond to symbols
that \emph{must} be equal, and different digits correspond to symbols that \emph{can} be different. And, indeed, there is short square substring $2323$ in the middle.
\end{example}
 
In general, it is not true that all strings have at most one long minimal square substring,
there are some counterexamples.
Indeed, the string $abcabcab$
has three distinct long minimal square substrings: $abcabc$, $bcabca$ and $cabcab$. 
However, all these squares are of the same half-length and start in the consecutive
positions of the original string, meaning that a lot of constraints actually coincide.
Intuitively, all counterexamples have to look in this, very regular, way.

The following Lemmas~\ref{square-overlap} and~\ref{key-lemma} are exact statements
that correspond to this intuition.

\begin{lemma}\label{square-overlap} Let $s$ be a string (over any alphabet), such that some its proper prefix $uu$ is a minimal
square and some its proper suffix $vv$ is a minimal square. Then, either $|s| \geqslant 3\min(|u|,
|v|) + 1$, or $|u| = |v|$ and $s = uu p$ for some non-empty prefix $p$ of $u$.
\end{lemma}
\begin{proof}[Sketch of an automatic ``partial proof''] Suppose that we want to check this lemma
for $|u| \leqslant d$ and $|v| \leqslant d$. Let us iterate on the length of $s$ (up to $3d$) and
create a graph with $|s|$ vertices, with edges corresponding to ``forced'' equalities between
symbols: edges between $i$ and $i + |u|$ for $0 \leqslant i < |u|$ and edges between
$i$ and $i + |v|$ for $|s| - 2|v| \leqslant i < |s| - |v|$. Now, the connected components in this graph
tell which symbols have to be equal and which do not. Now, we can just check whether there are
any small forced squares that disprove that either $uu$ or $vv$ was a \emph{minimal} square. 
The whole
procedure needs polynomial in $d$ time. Specifically, $O(d^5)$ for the most straightforward
implementation. Hence, running the above procedure can actually \emph{prove} the Lemma,
but only
for small lengths. I verified the Lemma for lengths of $|s|$ up to $200$ this way~\cite[test\_overlay.cpp]{github}.
\end{proof}
\begin{proof}[Mathematical proof] The mathematical proof is messier, but works for all lengths. 
Proof by contradiction. Let $s = s_0 s_1 \ldots s_{|s| - 1}$.

Consider the case $|u| = |v|$ first. Then, either $|s| \geqslant 3|u| + 1$ (and we are done), or $|s| \leqslant 3|u|$. In the latter case, we know $s_i = s_{|u| + i}$ for
$0 \leqslant i < |u|$, because $uu$ is a prefix of $s$ and $s_i = s_{|u| + i}$ for $|s| - 2|u| \leqslant i \leqslant
|s| - |u|$, because $vv$ is a suffix of $s$. Therefore, $s_i = s_{|u| + i}$ for $0 \leqslant i < |s| - |u|$, because $|s| - 2|u| \leqslant |u|$.  Hence,
$s_i = s_{i \bmod |u|}$ for $0 \leqslant i < |s|$, meaning that $s$ is a prefix of $u^*$. Because
the length of $u$ is between $2|u| + 1$ and $3|u|$, it has the same exact form as promised
by the lemma.

Now, suppose that $|u| \neq |v|$. Without loss of generality, $|u| < |v|$. Then suffix $rr$ overlaps with first $u$: otherwise the whole string $s$
has length at least $|u| + 2|v|$, which is at least
$|u| + 2(|u| + 1) = 3|u| + 2 \geqslant 3|u| + 1$, 
contradiction. Hence, $u = fg$ and $v = gh$, where $g$
is the non-empty overlap between the first $u$ and $vv$. Moreover, $f$ is non-empty,
because otherwise $uu$ would be substring of $vv$. Finally, because the right square is longer,
$|gh| = |v| > |u| = |fg|$, hence $|h| > |f|$.

Now we know almost everything about the relative positions of $uu$ and $vv$. More
specifically, $fgfg = uu$ is a prefix of $fghgh = fvv$. Hence, $u = fg$ is a prefix
of $hg$ (here we use that $|h| > |f|$). This, in turn, implies that $f$ is a proper prefix of $h$:
$h = fx$ for some non-empty string $x$. Therefore, because $fgfg = (fgf)(g)$ is a prefix 
of $fghgh = fg(fx)gh = (fgf)(xg)h$. Hence, $g$ is a prefix of $xg$. 

Suppose that $|g| \leqslant |x| - 1$. Then, $|s| \geqslant |fghgh| = |f(gfx)(gfx)| =
 2|x| + 2|g| + 3|f| \geqslant 4|g| + 2 + 3|f| = 3|fg| + |g| + 2 = 3|u| + |g| + 2 \geqslant
 3|u| + 1$. Contradiction.
  
Now we know that $|g| \geqslant |x|$.   This, along with $g$ being a prefix of $xg$, 
means that $x$ is a prefix of $g$: $g = xy$ for some, possibly empty, string $y$.
Hence, $vv = g(h)(g)h = g(fx)(xy)h = gf(xx)yh$ is not a minimal square, because of a square
substring $xx$. Contradiction.  
\end{proof}

\begin{lemma}\label{key-lemma} Let $s$ be a string that is not square free, but does not have
square substrings with half-length strictly less than $|s|/3$ (no rounding here). Then, $s$
has a unique inclusion-maximal substring of form $wwp$, where $ww$ is a minimal square and $p$ is some,
possibly empty, prefix of $w$. Moreover, $s$ has exactly $|p|+1$ minimal square substrings
~--- specifically the substrings of $wwp$ with length $2|w|$.
\end{lemma}

\begin{proof} The case when $s$ has exactly one minimal square substring corresponds to $p = \varepsilon$. Now, suppose that $s$ has at least two distinct minimal square substrings. Consider any two of them. Because they are minimal squares, neither of them is a substring of another. So,
one of them starts and ends earlier than another and we can apply Lemma~\ref{square-overlap}.

Hence, these squares have the same length (otherwise $|s| > 3(|s|/3) + 1 = |s| + 1$). Because the above statement is true for \emph{any} two minimal square substrings of $s$, \emph{all} minimal 
square substrings of $s$ have the same length. Consider the leftmost and the rightmost of them.
They intersect because $s$ is short enough, and their union has the form $wwp$ with $p$ being a prefix of $w$ by Lemma~\ref{square-overlap}. Then, any substring of their union with length $2|w|$
is a minimal square. Indeed,
all substrings of $wwp$ with length $2|w|$ are cyclic shifts of $ww$. A cyclic shift
of a minimal square is also a minimal square; one can prove this either by case analysis or
by using Shur's result that a square is minimal if and only if its half is square-free as a 
\emph{cyclic} string~\cite[Proposition 1]{circular}. 

Summarising, all substrings of $wwp$ with length $2|w|$ are minimal squares, and $s$
has no other minimal square substrings, because $wwp$ was chosen to be the union of 
the leftmost one and the rightmost one
\end{proof}

\begin{remark} In Lemmas~\ref{square-overlap} and~\ref{key-lemma} slightly better bounds are
actually true, but even the best possible bounds lead only to constant factor improvements in the final algorithm.
\end{remark}

Let $n$ be the length of square-strings we need to count.
Moreover, let $A$ be the automaton
from Theorem~\ref{crucian} for the set $M_{\prt{n/3}}$ and $f(\ell, q)$ be the number of 
strings with length $\ell$ that are rejected by $A$, when the computation starts in the
state $q$. In other words, $f(\ell, q)$ is the number of strings $s$, such that $|s| = \ell$ and
$\delta(q, s) \notin F$.

\begin{definition} A square is said to be \emph{short}, if its half-length is at most $\prt{n/3}$.
\end{definition}

\begin{definition} A string is \emph{promising} if it has no short square substrings.
\end{definition}

\begin{remark}\label{reverse-sqfree} A string $s$ is promising if and only if $s^R$ is promising: if $s$ contains square
$ww$, then $s^R$ contains $w^R w^R$ and vice versa.
\end{remark}

\begin{definition} For a promising string $t$ with length at most $n$ and an integer number $\ell$ with
$0 \leqslant \ell \leqslant n - |t|$, denote by $g(\ell, t)$ the number of promising strings of length
$n$ that have $t$ as a substring, starting with position $\ell$. 
\end{definition}

\begin{lemma}\label{decent-with-substring} For any promising string $t$ with $2\prt{n/3} + 1 \leqslant |t| \leqslant n$ and any integer $0 \leqslant \ell \leqslant n - |t|$, $g(\ell, t) = f(\ell, \delta(q_0, t^R)) \cdot f(n - |t| - \ell, \delta(q_0, t))$.
\end{lemma}
\begin{proof} Suppose that $xty$ is a string satisfying the condtions with lemma, with $|x| = \ell$.
Then, $xty$ is promising if and only if $xt$ and $ty$ are both promising. Indeed, because $|t| \geqslant 2\prt{n/3}+1$, any
possible short square in $xty$ fully fits in either $xt$ or $ty$. By definition of function $f$, $ty$ is promising if and only if $A$ rejects $y$, starting from $\delta(q_0, t)$. Hence, there is 
$f(|y|, \delta(q_0, t)) = f(n - |t| - \ell, \delta(q_0, t))$ ways to choose $y$. 

Similarly, by Remark~\ref{reverse-sqfree}, $xt$ is promising if and only if $t^R x^R$ is promising.
Hence, there are $f(|x^R|, \delta(q_0, t^R)) = f(\ell, \delta(q_0, t^R))$ ways to choose $x$.
All in all, there are $f(\ell, \delta(q_0, t)) \cdot f(n - |t| - \ell, \delta(q_0, t))$ ways to choose
$x$ and $y$.
\end{proof}

Now, everything is in line for the improved algorithm.

\begin{theorem} One can compute $|L_n|$ in $O^*(|L_{\prt{n/2}}|)$ time and $O^*(|M_{\prt{n/3}}|)$
memory. 
\end{theorem}
\begin{proof} By Lemma~\ref{key-lemma}, there are three types of strings with length $n$:
\begin{enumerate}
\item[1)] not promising
\item[2)] promising, but not square-free, they have $wwp$ substring as per Lemma~\ref{key-lemma}.
\item[3)] square-free
\end{enumerate}

We want to know the number of strings of type 3.
By definition of being promising, the total number of strings of types 2 and 3 is $f(n, q_0)$.

Consider any string of type 2. We can try to count them using Lemma~\ref{decent-with-substring}, by iterating over a minimal square substring and its position.
Of course, there is massive overcounting happening here: if $wwp$ is the
substring of $s$ that is given by Lemma~\ref{key-lemma}, then we count $s$ exactly
$|p| + 1$ times. To deal with this, notice that, for such a string, there are exactly $|p|$
substrings of type $xxx_0$, where $xx$ is a minimal square: exactly the substrings of  $wwp$
with length $2|w|+1$. Hence, counting them with minus sign fixes the overcounting problem perfectly,
because $(|p|+1) - |p| = 1$.

In the end, there are 
\begin{equation}\label{main_expression}
f(n, q_0) - \left( \sum\limits_{ww} \sum\limits_{i=0}^{n-|ww|} g(i, ww)
- \sum\limits_{ww} \sum\limits_{i=0}^{n-|ww|-1} g(i, www_0) \right)
\end{equation}
strings of type 3,
where both summations are over minimal squares with half-length at least $\prt{n/3}+1$.

Let's trace the steps necessary to complete the algorithm:
\begin{enumerate}
\item[1)] Find the set $M_{\prt{n/3}}$ and build the automaton $A$. Takes
$O^*(L_{\prt{n/3}})$ time and $O^*(1)$ memory.
\item[2)] Compute the values $f(\ell, q)$ for $0 \leqslant \ell \leqslant \prt{n/3}$ and $q \in Q$. 
Takes $O^*(M_{\prt{n/3}})$ time and memory.
\item[3)] Iterate over all minimal squares with half-length at most $\prt{n/2}$, in order to 
compute the sum~\eqref{main_expression}. This is the slowest
part. Iterating over all minimal squares with half-length at most $\prt{n/2}$ takes $O^*(L_{\prt{n/2}})$ time and $O^*(1)$
memory. Notice, that there is no need to actually store them all in memory, knowing
only the current one and the values of $f$ is enough. Lemma~\ref{decent-with-substring}
comes in play here, allowing to express $g$'s through $f$'s.
\end{enumerate}

In the end, total time complexity is still $O^*(L_{\prt{n/2}})$, but the memory complexity
is $O^*(M_{\prt{n/3}})$, as promised.
\end{proof}

\section{Possible time-memory tradeoffs?}\label{improvement}

I like to think about the algorithm from Section~\ref{main_algo} as an incomplete application of the inclusion-exclusion principle.
Indeed, we take all promising strings of length $n$ and subtract promising strings with at least one minimal square substring
(well, up to technical details in form of the $wwp$ substrings). 
In a normal inclusion-exclsuion algorithm, we would need to 
add back promising strings with at least two different minimal squares, then subtract promising strings with at least three different
minimal squares again, et cetera. However, it turns out that, up to some simple counterexamples, 
there are \emph{no}
promising strings with two different minimal squares! 

But what will happen if we replace $\prt{n/3}$ with a smaller number, say, $\prt{n/10}$, and do several steps
of inclusion-exclusion instead of just one? 

As it turns out, this leads to smaller memory consumption at the cost of higher running time.
Indeed, let's fix some $k \geqslant 4$.
\begin{definition} A square is said to be \emph{$k$-short}, if its half-length is at most $\prt{n/k}$.
\end{definition}
\begin{definition} A string is \emph{$k$-promising} if it has no $k$-short square substrings.
\end{definition}
Consider any $k$-promising string $s$. Intuitively, Lemma~\ref{key-lemma} implies that all pairs of minimal square substrings of $s$ either have small intersection or are both a part of a large $wwp$ block.
Hence, there ought to be only $O(1)$ such blocks~--- otherwise some would have large intersection
by Dirichlet's principle.

Let us explain the intuition from the previous paragraph formally. 
Consider \emph{all} minimal
square substrings of some $k$-promising $s$, sorted by the coordinate of their left end: $s[\ell_1, r_1), s[\ell_2, r_2), \ldots, s[\ell_d, r_d)$, with $\ell_1 < \ell_2 < \ldots < \ell_d$ (all inequalities are strict; otherwise
some minimal square would be a prefix of another). Then, $r_1 < r_2 < \ldots < r_d$ (otherwise some minimal square is a substring of another). For each $i$ from $1$ to $d$ inclusive, 
denote the middle position of the $i$-th minimal square by $m_i$. In other words, $m_i = (\ell_i + r_i)/2$. It is easy to see that middles are also increasing: if $\ell_i < \ell_{i+1}$, but $m_i 
\geqslant m_{i+1}$, then $r_i = 2 \cdot m_i - \ell_i > 2 \cdot m_{i+1} - \ell_{i+1} = r_{i+1}$. Finally, denote the square itself by $u_i u_i$. That is, $s[\ell_i, m_i) = s[m_i, r_i) = u_i$.

Indeed, consider some index $1 \leqslant i \leqslant d - 1$.
Then, by Lemma~\ref{key-lemma}, there are the following possibilities:
\begin{enumerate}
\item[1.] Substrings $s[\ell_i, r_i)$ and $s[\ell_{i+1}, r_{i + 1})$ do not intersect at all. In other words,
$r_i \leqslant \ell_{i + 1}$. Then, $m_{i + 1} = \ell_{i + 1} + |u_{i+1}| > \ell_{i+1} + \prt{n/k} \geqslant r_i + \prt{n/k} = m_i + |u_i| + \prt{n/k} > m_i + 2\prt{n/k}$ (the first and the last inequalities corresponds to the fact that all square substrings of $s$ are long enough).
\item[2.] Substrings $s[\ell_i, r_i)$ and $s[\ell_{i+1}, r_{i+1})$ intersect, but the length of their union, the substring $s[\ell_i, r_{i+1})$, is at least $3 \min(|u_i|, |u_{i+1}|) + 1$. That is,
\begin{equation}\label{from-lemma}
r_{i + 1} - \ell_i \geqslant 3\min(|u_i|, |u_{i+1}|) + 1
\end{equation}
Because $s[\ell_i, r_{i+1})$ contains both $u_i u_i$ and $u_{i+1} u_{i+1}$ as proper substrings, $(r_{i+1} - \ell_i) \geqslant 2\max(|u_i|, |u_{i+1}|)$. By taking the average with inequality~\eqref{from-lemma}, $r_{i+1} - \ell_i \geqslant
(3\min(|u_i|, |u_{i+1}|) + 2\max(|u_i|, |u_{i+1}|)/2 + 1 = \min(|u_i|, |u_{i+1}|)/2 + ( \min(|u_i|, |u_{i+1}|) +  \max(|u_i|, |u_{i+1}|)) + 1 =  \min(|u_i|, |u_{i+1}|)/2 + (|u_i| + |u_{i+1}|) + 1 \geqslant \prt{n/k}/2 + (|u_i| + |u_{i+1}| + 1)$. Hence, $m_{i + 1} - m_i = (r_{i+1} - \ell_i) - (|u_i| + |u_{i+1}|) \geqslant \prt{n/k}/2 + 1$.
\item[3.] The string $s[\ell_i, r_{i+1})$ has small length ($r_{i+1} - \ell_i \leqslant 3\min(|u_i|, |u_{i+1}|)$), but $|u_i| = |u_{i+1}|$. Then, by the conclusion of Lemma~\ref{key-lemma}, 
$s[\ell_i + 1, r_i + 1)$ is a minimal square. Therefore, $\ell_{i+1} = \ell_i + 1$ and $r_{i+1} = r_i + 1$. In this case, the difference between $m_{i+1}$ and $m_i$ is not large, but,
similarly to the Section~\ref{main_algo}, we can consider such minimal squares in batches.
\end{enumerate}
Hence, all minimal square substrings of $s$ split into $b$ inclusion-maximal batches for some $b \geqslant 0$, with $i$-th ($1 \leqslant i \leqslant b$) of them defined by three parameters $L_i$, $R_i$ and $T_i \geqslant 1$:
the first minimal square in the batch and the size of the batch. Formally speaking, a batch $(L_i, R_i, T_i)$ corresponds to the fact that substrings $s[L_i + j, R_i + j)$ are minimal squares for each $0 \leqslant j < T_i$, but $R_i + T_i > |s| = n$ or $s[L_i + T_i, R_i + T_i)$ is not a minimal square and, similarly, $L_i - 1 < 0$ or $s[L_i - 1, R_i - 1)$ is not a minimal square.

Let $M_i = (L_i + R_i)/2$ be the middle of the first square in each batch. From the above, it follows that $M_i$'s are increasing rather quickly. More specifically, $M_{i+1} - M_i > \prt{n/k}/2$ for each $1 \leqslant i \leqslant b - 1$. Hence, $b \leqslant 2k + 1$~--- otherwise $M_b > (2k+1-1)\cdot\prt{n/k}/2 \geqslant n$.

Hence, for any $k$-promising string, there are $O(k)$ batches in total. Each batch is uniquely defined by its integer parameters $(L_i, R_i, T_i)$ and a square-free string $s[L_i, M_i)$ of length $M_i - L_i = (R_i - L_i)/2$. Of course, some square-free strings do not correspond to a valid batch, but this it not important right now. From now on, by \emph{batch}, I mean the tuple 
$(L_i, R_i, T_i, U_i)$, with $U_i$ being a square-free string of length $(R_i - L_i)/2$. 
A string $s$ \emph{contains} a batch $(L_i, R_i, T_i, U_i)$ if $s[L_i, L_i + |U_i|) = U_i$, substrings $s[L_i + j, R_i + j)$ are minimal squares for each $0 \leqslant j < T_i$ and the batch itself is maximal possible by inclusion (in other words, $L_i - 1 < 0$ or $s[L_i - 1, R_i - 1)$ is not a minimal 
square and $R_i + T_i > n$ or $s[L_i + T_i, R_i + T_i)$ is not a minimal square).

\begin{example} For $k = 4$, a string $abcabcabc$ is $k$-promising and contains
exactly one batch: $(0, 6, 4, abc)$. It does not contain batches $(1, 7, 3, bca)$, $(0, 6, 3, abc)$ and $(1, 7, 2, bca)$, because they are not inclusion-maximal.
\end{example}

We want to compute the number of square-free strings, or, in other words, $k$-promising strings that contain no batches. For a set $S$ of batches let $h(S)$ be the number of $k$-promising
strings of length $n$ that contain all batches from the set $S$, but \emph{may also contain some other batches}. Then, by inclusion-exclusion, the answer for length $n$ is
$\sum\limits_{|S| \leqslant 2k + 1} (-1)^{|S|} h(S)$, where the summation is over all possible sets of batches (as we know already, there are no strings that contain $2k + 2$ or more batches). Hence, we are left with the two following subproblems:
\begin{itemize}
\item[1.] For a given set $S$ of batches, compute $h(S)$ quickly enough.
\item[2.] Iterate over all possible sets of batches efficiently. In particular, prove that there are not 
too many possible sets.
\end{itemize}

Let us solve the first subproblem first.
\begin{lemma}\label{single} $h(S)$ can be computed in $O^*(2^{O(k)})$ after precomputation
that uses $O^*(|M_{\prt{n/k}}|^3)$ time and $O^*(|M_{\prt{n/k}}|^2)$ memory. 
Moreover, for $k = 4$,  
$h(S)$ can be computed in $O^*(2^{O(k)}) = O^*(1)$ time after precomputation that uses $O^*(|M_{\prt{n/k}}|)$ time and memory. 
\end{lemma}
\begin{proof} Suppose that some string $s$ contains every batch from $S$. 
Then, we already know what some symbols in $s$ are equal to. 
Moreover, because each batch from $S$ is inclusion-maximal, we know for some symbols what 
they are \emph{not} equal to. Specifically, for a batch $(L, R, T, U)$ and $M \coloneqq (L+R)/2$ we know that $s_{L - 1} \neq s_{M-1}$ and $s_{R+T-1} \neq s_{M+T-1}$ (if $L - 1 \geqslant 0$ and $R + T - 1 < n$ respectively, of course). For each such symbol (at most $2 \cdot (2k + 1) = O(k)$ of them), iterate over two possibillities. For each of those $2^{O(k)}$ cases, 
check two things (both can be done in $O^*(1)$ time by simply iterating over all fully-known substrings of $s$):
\begin{itemize}
\item that $s$ does not contain a $k$-short square consisting only of known symbols,
\item that each batch from $S$ indeed is a valid batch contained in $s$.
\end{itemize}

Now, we are left with a simpler problem: how many $k$-promising strings are there, assuming
that symbols on some positions are already known? Moreover, positions with known symbols appear
in blocks of length at least $2(\prt{n/k}+1)$ each. Hence, any $k$-short square intersects exactly one 
block of \emph{unknown} symbols (otherwise it fully contains a block of known symbols and, therefore, cannot be $k$-short).

Firstly, let us deal with the simpler case of $k = 4$. In this case, there is at most one block of 
known symbols. Indeed, each such block has length at least $2(\prt{n/4} + 1)$ and there is just not enough space for two of them. Hence, there is an unknown prefix, a fully-known middle and an unknown suffix (each of those three parts may be empty). What we need to know is the number
of $k$-promising strings that conform to this pattern. This situation already appeared before: specifically, see Lemma~\ref{decent-with-substring}. We can define and compute the functions $f(\cdot, \cdot)$ and $g(\cdot, \cdot)$ in the same way, with only difference being that the automaton we build will corespond to $k$-short squares and will therefore have size $O^*(|M_{\prt{n/k}}|)$.

In the general case, there \emph{may} be blocks of unknown symbols that are surrounded by known symbols from both left and right. However, all blocks of unknown symbols can still be filled independently. Consider a block of unknown symbols of length $\ell$ that is surrounded by (possibly, empty) blocks $w_p$ and $w_s$ of known symbols. Let $A_k = (Q, q_0, \delta, F)$ be the automaton from theorem~\ref{crucian} for $k$-short squares. Then, we can fill-in unknown symbols with a string $s \in \Sigma^{\ell}$ if and only if $\delta(q_0, w_p s w_q) \notin F$. In other words,
$\delta(\delta(q_0, w_p), s), w_q) \notin F$. 

Hence, let us compute $f_{\mathrm{both}}(\ell, p, q)$: how many strings $s \in \Sigma^\ell$ are there, such that $\delta(s, p) = q$. We can do this in $O^*(|M_{\prt{n/k}}|^2)$ by dynamic programming over the states of $A_k$. To compute the number of ways to fill the block, 
substitute $p \coloneqq \delta(q_0, w_p)$ and iterate over all $q$, such that $\delta(q, w_q) \notin F$. 

Unfortunately, this approach takes $O^*(2^{O(k)} \cdot |M_{\prt{n/k}}|)$ time to compute $h(S)$ (the second factor comes from iterating over $q$). 
To get rid of the second factor, notice the following: for \emph{any} $s$, whether or not $w_p s w_q$ has
any $k$-short square substrings, depends only on $\delta(q_0, w_p)$ and $\delta(q_0, w_q^R)$, but not on their exact values (this immediately follows from the Theorem~\ref{crucian} and
the fact that $w_p s w_q$ does not have any $k$-short squares if and only if $(w_p s w_q)^R = w_q^R s^R w_p^R$ also does not. 

Hence, the numbers of ways to fill-in the block depends only on its length, $\delta(q_0, w_p)$ and $\delta(q_0, w_q^R)$. We can simply precompute all those $O^*(|M_{\prt{n/k}}|^2)$ numbers,
each in $O^*(|M_{\prt{n/k}}|)$ time.
\end{proof}

Finally, we need to iterate over all possible sets of batches somehow. Iterating over 
the numbers $L_i, R_i, T_i$ takes only $O(n^{O(k)})$ time. To iterate over possible strings
$U_i$, iterate over batches from left to right and fill them that in that order. 
Because each batch consists of consecutive minimal squares, $L_i + T_i \leqslant L_{i+1}$ and $R_i + T_i \leqslant R_{i+1}$
for consecutive batches. 
Hence, for each batch, some prefix of $U_i$ is already known, and some, possibly empty, suffix is not. 
The unknown part is a square-free string by itself.
Moreover, each symbol in the unknown part corresponds to at least two positions in the string (otherwise we would have figured out this symbol already). 
Hence,
we need to iterate over $O(k)$ strings of total length at most $\prt{n/2}$. It is known that $|L_\ell|$ grows exponentially.
In particular, $c_1 \gamma^\ell \leqslant |L_\ell| \leqslant c_2 \gamma^\ell$ for some $\gamma$ and $c_1, c_2 > 0$. 
Hence, there are at most $c_2^k \gamma^{\prt{n/2}}$ ways to choose these strings, which is at most $|L_{\prt{n/2}}| \cdot c_2^k / c_1
= O(|L_{\prt{n/2}}| \cdot 2^{O(k)})$. In total, iterating over all possible sets $S$ takes $O(|L_{\prt{n/2}}| \cdot n^{O(k)})$ time.

Hence, we need $O^*(|M_{\prt{n/k}}|^2)$ memory and $O^*(|L_{\prt{n/2}}| \cdot n^{O(k)})$ time (precomputation from Lemma~\ref{single} is irrelevant for large $k$).
Moreover, for $k = 4$ only $O^*(|M_{\prt{n/k}}|)$ memory is needed. Unfortunately, the practical value of this optimisation is questionable. The memory consumption of the algorithm from 
the Section~\ref{main_algo} is, indeed, quite a problem already for $n = 141$, but adding even \emph{one} extra $O(n)$ factor to the time complexity turns ``several hours'' into ``several \emph{weeks}''. Moreover, assuming that $|M_\ell|$ grows exponentially with $\ell$, we need to choose either
$k = 4$ or $k \geqslant 7$ to get any memory advantage. Because of the above, choosing $k \geqslant 7$ is completely hopeless. Choosing $k = 4$ is  an interesting idea that may lead to
a better results in the end, but I have not implemented it yet.

The main running time bottleneck of this approach is pretty apparent: even when $(L_i, R_i, T_i)$ are fixed, 
I do not know any way to avoid iterating over almost a half of the whole string in the worst case.
In fact, it seems difficult to compute the number of minimal squares of half-length $n$ in significantly less than $O(|M_n|)$ time. Intuitively,
counting only minimal squares corresponds to the first step of inclusion-exclusion and should therefore be easier somehow. However, even such, intuitively simpler, problem seems to be out of reach now.

\section{Final notes}

The algorithm from Section~\ref{main_algo} is implemented in the linked repository~\cite{github}, with some constant optimisations and other minor tweaks. There are still several optimisations
possible, both in terms of time and memory, but they are more annoying to implement. If you want to suggest some code improvements, contact me via e-mail.

As noted in the Section~\ref{improvement}, any substantial improvement to counting square-free words would likely require a faster way to count minimal squares. 
I believe that it also works in the opposite direction: any non-trivial algorithm for counting
minimal squares will lead to a better algorithm for counting square-free words.

\end{document}